\documentclass[journal,twoside,web]{ieeecolor}
\usepackage{generic}

\let\labelindent\relax

\usepackage{amsmath,amsthm,amssymb,paralist,subfigure,cite,graphicx,color,amsbsy,float}
\usepackage{stmaryrd,mathrsfs,url}
\usepackage[table]{xcolor}
\usepackage{cuted,mathtools,lipsum}
\usepackage[ruled,vlined,linesnumbered]{algorithm2e}
\usepackage{soul}

\usepackage{pifont}
\usepackage[utf8]{inputenc} 
\usepackage{hyperref}       
\hypersetup{
	colorlinks=true,
	linkcolor=cyan,
	filecolor=mnodea,
	urlcolor=cyan,
	citecolor=lime,
}
\usepackage{url}            
\usepackage{booktabs}       
\usepackage{amsfonts,amsmath,amssymb}       
\usepackage{amsthm}     
\usepackage{nicefrac}       
\usepackage{microtype}      
\usepackage{graphicx}
\usepackage{dsfont}

\newtheorem{lem}{Lemma}
\newtheorem{ass}{Assumption}
\newtheorem{theorem}{Theorem}

\def\mb{\mathbf}
\def\mbb{\mathbb}
\def\mc{\mathcal}

\DeclareMathOperator*{\argmin}{argmin}

\begin{document}
	\title{\bf \LARGE Logarithmically Quantized Distributed Optimization over Dynamic Multi-Agent Networks}
	
	\author{Mohammadreza Doostmohammadian, Sérgio Pequito
		\thanks{		
			M. Doostmohammadian is with the Mechatronics Department, Faculty of Mechanical Engineering, Semnan University, Semnan, Iran {\texttt{doost@semnan.ac.ir}}.
			 S. Pequito is with Department of Electrical and Computer Engineering and Institute for Systems and Robotics, Instituto Superior T\'ecnico, University of Lisbon, Portugal {\texttt{sergio.pequito@tecnico.ulisboa.pt}}.}}
	\maketitle
	\pagestyle{empty} 
	\thispagestyle{empty} 
	\begin{abstract}
		Distributed optimization finds many applications in machine learning, signal processing, and control systems. In these real-world applications, the constraints of communication networks, particularly limited bandwidth, necessitate implementing quantization techniques. 
   In this paper, we propose distributed optimization dynamics over multi-agent networks subject to logarithmically quantized data transmission. Under this condition, data exchange benefits from representing smaller values with more bits and larger values with fewer bits. As compared to uniform quantization, this allows for higher precision in representing near-optimal values and more accuracy of the distributed optimization algorithm.
  		The proposed optimization dynamics comprise a primary state variable converging to the optimizer and an auxiliary variable tracking the objective function's gradient.
		Our setting accommodates dynamic network topologies, resulting in a hybrid system requiring convergence analysis using matrix perturbation theory and eigenspectrum analysis. 
	\end{abstract}

 \vspace{-0.3cm}
	\section{Introduction} \label{sec_intro}
	\IEEEPARstart{D}{istributed} optimization and learning have gained significant interest due to advancements in cloud computing, parallel processing, and the Internet of Things  \cite{ren2010distributed}. In real-world applications, the constraints of communication networks, particularly limited bandwidth, necessitate implementing quantization techniques. 
 Unfortunately, most existing distributed optimization and learning algorithms do not consider real-world constraints like quantization and saturation. Examples include dual-based decomposition \cite{terelius2011decentralized}, one-bit diffusion-based learning \cite{zayyani2018bayesian}, optimization under noisy information sharing \cite{pan2024utilizing,wang2022gradient,reisizadeh2022distributed,zhang2018distributed}, accelerated convergence by using the history information to achieve fast and accurate estimation of the average gradient \cite{qu2017harnessing} or by using signum-based non-Lipschitz functions \cite{spl24}, and Lagrangian dual-based algorithms via edge-based methods \cite{shi2018augmented} -- see \cite{yang2019survey,nedic2018distributed} for surveys.

Logarithmic quantization balances communication efficiency and solution accuracy, preserving important information while minimizing communication overhead \cite{dimakis2010gossip,lee2017lognet,bu2017stability}. It often outperforms linear (uniform) quantization in implementation, replacing multiplier hardware with shifters or adders \cite{oh2021automated}. This approach is particularly beneficial in several scenarios: (\emph{i}) when data distribution across agents is non-uniform or skewed \cite{lee2017lognet}; (\emph{ii}) when variables have a wide dynamic range, common in distributed optimization problems; (\emph{iii}) when dealing with neural network training, where weights are distributed non-uniformly around zero and activation functions are highly concentrated near zero \cite{gholami2022survey}; (\emph{iv}) when encoding data with large dynamic ranges in fewer bits than fixed-point representation \cite{miyashita2016convolutional}; and (\emph{v}) when implemented in digital hardware, where log-domain representation is natural~\cite{jiang2024hardware}. 

Logarithmic quantization in distributed resource allocation and scheduling via equality-constraint optimization was addressed in \cite{ojsys,my_ecc}. Distributed resource allocation under quantized agents with saturation was considered in \cite{ccta}. Other related work includes distributed optimization with linear (uniform) quantization \cite{rabbat2005quantized,lee2018finite,lee2019deep,rikos2023distributed,rikos2021optimal}. 
Our research addresses the gap in the literature regarding logarithmic quantization over consensus-constraint distributed optimization networks. Unlike uniform quantization, this approach leads to no optimization residual and optimality gap. Motivating applications include networked empirical risk minimization \cite{xin2020decentralized}, distributed logistic regression \cite{qureshi2021decentralized}, distributed support vector machine (SVM) \cite{dsvm}, and classification via neural networks \cite{qureshi2020s}. Note that distributed optimization problems often involve variables with a wide dynamic range. Logarithmic quantization adapts well to such scenarios by allocating more bits to the informative lower end  (near the optimal point) of the range and fewer bits to the less informative higher end. This is an improvement over uniform quantization with fixed non-adaptive bit allocation.

The main contributions of our work are as follows:
\textbf{1.} We propose distributed dynamics for logarithmic quantization in data networks, introducing an auxiliary variable to track gradients and update node states;
\textbf{2.} We analyze the convergence of gradient tracking dynamics over \textit{dynamic} networks using a perturbation-based method to prove convergence for sufficiently small update rates;
\textbf{3.} We apply our approach to distributed SVM-based binary classification, demonstrating its practical utility in cooperative data classification over quantized dynamic networks;
\textbf{4.} Our logarithmically quantized dynamics achieve exact convergence, surpassing linear uniform quantization methods with optimality gaps.

	\section{Problem Formulation and Assumptions} \label{sec_prob}
	The distributed optimization problem is formulated as minimizing the sum of local objective functions over a \mbox{multi-agent} network with $n$ agents, i.e., 
	\begin{align}\nonumber
		\min_{\mb{x} \in \mathbb{R}^{nm}} &
		F(\mb x) = \sum_{i=1}^{n} f_i(\mb{x}_i)\\\label{eq_prob}
		\text{subject to}& ~ \mb{x}_1 = \mb{x}_2 = \dots = \mb{x}_n,
	\end{align}
	where $\mb{x}_i\in\mathbb{R}^m$ is the state of node/agent $i$ and column vector $\mb{x}:=(\mb{x}_1;\dots;\mb{x}_n)$ as the global state. The local objective function at agent/node $i$ is
	$f_i:\mathbb R^{m} \mapsto \mathbb R$ and $F:\mathbb R^{nm} \mapsto \mathbb R$ is the global objective. 
 Note that the coupling between the objective functions $f_i$ is moved from the decision variable to the consistency constraints. 
  In the present setting, we seek to solve problem~\eqref{eq_prob} in a distributed way. This implies that the learning and computations are local and every agent only uses the data in its neighborhood to solve the problem, i.e.,  to determine ${\mb{x}}^* \in \mathbb{R}^{nm}$ which is in the form~$\mb x^* := [\mb{x}^*_1;\mb{x}^*_2;\dots;\mb{x}^*_n] =  \mb{1}_n \otimes \overline{\mb{x}}^*$ and ${\sum_{i=1}^{n} \nabla f_i(\mb{x}^*_i) = \mb{0}_m}$, where $\overline{\mb{x}}^*$ is the optimal state at each local node, and $\mb{1}_l$ and $\mb{0}_l$  denote all ones and all zeros column vectors of size $l$, respectively.

Additionally, we require the following technical assumptions.
	 
\begin{ass} \label{ass_cost}
 The functions $f_i:\mathbb R^{m} \mapsto \mathbb R$ with smooth gradients for $i\in\{1,\ldots,n\}$ might be locally non-convex, but satisfy
 \begin{equation} \label{eq_sum_df}
 	(\mb{1}_n \otimes I_m)^\top H  ( \mb{1}_n \otimes I_m)= \sum_{i=1}^n \nabla^2  f_i(\mb{x}_i) \succ 0,
 \end{equation}
 with $H:=\mbox{diag}[\nabla^2 f_i(\mb{x}_i)]$ and there exists $\gamma>0$ for which $\nabla^2 f_i(\mb{x}_i) \preceq \gamma I_m$,   $I_m$ denotes the $m\times m$ identity matrix.\hfill $\circ$
\end{ass} 
This assumption ensures that the global objective function has no local minimum and our gradient-based dynamics converges to the global minimum of the overall objective.
In real-world distributed systems, the communication topology between agents is inherently dynamic, fluctuating due to factors such as node mobility, varying channel conditions, or intermittent connectivity. 
The time-varying network structure is mathematically represented by a directed graph (digraph) $G_\theta(\mc{V}_\theta, \mc{E}_\theta \subset \mc{V}_\theta \times \mc{V}_\theta)$, where $\mc{V}_\theta$ is the set of vertices (agents) and $\mc{E}_\theta$ is the set of edges at configuration $\theta$. An edge $(i,j) \in \mc{E}_\theta$ denotes that agent $i$ sends information to (i.e., communicates with) node $j$. The network's dynamic nature is captured by a switching signal $\theta \in \{1, \ldots, \mc{M}\}$, reflecting $\mc{M}$ possible network configurations.


 \begin{ass} \label{ass_net}
 	 The  adjacency matrix~${W_\theta=\{w_{ij}^\theta\}}$, with $w_{ij}^{\theta} \in\mathbb{R}$,  associated with $\mc{G}_\theta$ is weight-balanced (WB), i.e., $\sum_{i=1}^n w_{ij}^\theta = \sum_{j=1}^n w_{ij}^\theta$ for every $\theta$. The associated Laplacian matrix is~$\overline{W}_\theta=\{\overline{w}_{ij}^\theta\}$ with entries $\overline{w}_{ij}^\theta=w_{ij}^\theta$ for $i\neq j$ and $\overline{w}_{ij}^\theta=-\sum_{i=1}^n w_{ij}^\theta$ for $i=j$. Furthermore, we need the network $\mc{G}_\theta$ to be connected at all times and for all $\theta$. \hfill $\circ$
 \end{ass} 


An example consensus-based WB setup is given in \cite{nowzari2016distributed}. Additionally, as agents collaborate to solve problem \eqref{eq_prob}, the volume of data exchanged poses significant challenges in \mbox{real-world} applications, particularly regarding network bandwidth. 
In summary, we seek to address the following problem:

\underline{\textbf{Problem 1}} Determine a solution to the  distributed optimization problem \eqref{eq_prob},  under the following constraints:
\begin{enumerate}
\item  \emph{Dynamic network topology:} The problem is set in a time-varying network structure, represented by a directed graph $G_\theta(\mc{V}_\theta, \mc{E}_\theta \subset \mc{V}_\theta \times \mc{V}_\theta)$, where $\theta \in \{1, \ldots, \mc{M}\}$ is a switching signal reflecting $\mc{M}$ possible network configurations.

\item \emph{Local computation:} Each agent must solve the problem using only local information and data from its neighborhood without access to global information.

\item \emph{Communication constraints:} The solution must address real-world challenges in data exchange, particularly regarding network bandwidth limitations, without losing optimality.
\end{enumerate}

 To address the aforementioned problem, we propose a distributed optimization framework leveraging quantization, a crucial technique for reducing data precision. This will facilitate efficient communication and enhance algorithm execution in real-time scenarios; specifically, we will deploy logarithmic quantization to achieve robust performance across dynamically changing network configurations. 

\section{Logarithmically Quantized Distributed Optimization over Dynamic Multi-Agent Networks} \label{sec_dyn}

\subsection{The Proposed Distributed Log-Quantized Dynamics} \label{sec_dyn_sol}

The proposed solution to \textbf{Problem 1} involves two dynamics: one to derive the agents to reach agreement on the optimization variable and one to track the consensus on the local gradients. For the latter, we introduce an auxiliary variable $\mb{y}$ to track the consensus on gradients, which derives the first dynamics toward the optimum. Then,
each agent performs the following:

\begin{align} \label{eq_xdot_g}	
	\dot{\mb{x}}_i &= -\sum_{j=1}^{n} w_{ij}^\theta (q(\mb{x}_i)-q(\mb{x}_j))-\alpha \mb{y}_i, \\ \label{eq_ydot_g}
	\dot{\mb{y}}_i &= -\sum_{j=1}^{n} w_{ij}^\theta (q(\mb{y}_i)-q(\mb{y}_j) ) + \partial_t \boldsymbol{\nabla} f_i(\mb{x}_i),
\end{align}
where  $\partial_t$ denotes the time-derivative $\frac{d}{dt}$, $w_{ij}^{\theta} \in\mathbb{R}$ is the weight associated with the edge $(i,j)$  over the graph $\mc{G}_\theta$, ~$\mb{x}_i$ as the state of node~$i$, the auxiliary variable ${\mb{y}=[\mb{y}_1;\mb{y}_2;\dots;\mb{y}_n]\in\mbb{R}^{mn}}$, $q(\mb{x}_i),q(\mb{y}_i)\in\mbb{R}^{mn}$ as the log-quantized versions of $\mb{x}_i,\mb{y}_i$, and the real scalar~$\alpha>0$ denotes the step-size for gradient tracking. 
Additionally, the function $q(\cdot)$ denotes the \emph{log-scale quantization}, implying that the gradient information shared over the network is logarithmically quantized. This log-scale quantization is formulated as \cite{guo2013consensus},
\begin{align}
	q(z) &= \mbox{sgn}(z) \exp\Big(\rho \left[ \frac{\log(|z|)}{\rho}\right]\Big), \label{eq_quan_log}
\end{align}
where $[\cdot]$ denotes rounding to the nearest integer, $\mbox{sgn}(.)$ is the sign function, and the log quantization level $\rho>0$ denotes a scaling parameter for the desirable quantization. It is worth mentioning that log quantization $q(\cdot)$ is a strongly sign-preserving odd mapping and it is sector-bounded as $(1-\frac{\rho}{2})z\leq q(z) \leq (1+\frac{\rho}{2})z$.

\subsection{Convergence of  Distributed \mbox{Log-Quantized} Dynamics} \label{sec_dyn_conv}

From the consensus structure of the proposed dynamics and sector-bound condition of log-scale 
quantization, by initializing $\mb{y}(0)=\mb{0}_{nm}$ both summations in the right-hand-side of the equations~\eqref{eq_xdot_g}-\eqref{eq_ydot_g} asymptotically go to zero, leading to
\begin{align} \label{eq_sumxdot}
	\sum_{i=1}^n \dot{\mb{y}}_i  = \sum_{i=1}^n \partial_t \boldsymbol{ \nabla} f_i(\mb{x}_i), ~~~
	\sum_{i=1}^n \dot{\mb{x}}_i
	= -\alpha \sum_{i=1}^n\mb{y}_i.
\end{align}
Then, by simple integration it readily follows that $\sum_{i=1}^n \dot{\mb{x}}_i $ tracks $ - \sum_{i=1}^n \boldsymbol{ \nabla} f_i(\mb{x}_i)$, which shows the gradient-tracking nature of the proposed dynamics. Also, from the consensus structure of the dynamics and the optimization problem, the equilibrium is ${\mb{x}}^*= \mb{1}_n \otimes \overline{ \mb{x}}^* $. To find the equilibrium ${\mb{x}}^*$, we set  ${\dot {\mathbf x}_i = \mb{0}_m}$ and ${\dot {\mathbf y}_i = \mb{0}_m}$; then, we have~${(\mathbf 1_n^\top \otimes I_m) \boldsymbol{ \nabla} F(\mb{x}^*) = \mb{0}_m}$,  $\dot{\mb{y}}_i = \frac{d}{dt} \boldsymbol{ \nabla} f_i(\overline{ \mb{x}}^*)=  \boldsymbol{ \nabla}^2 f_i(\overline{ \mb{x}}^*) \dot{\mb{x}}_i = \mb{0}_m$ and, consequently, $\sum_{i=1}^n \dot{\mb{x}}_i = -\alpha (\mathbf 1_n^\top \otimes I_m) \boldsymbol{ \nabla} F(\mb{x}^*) =  \mb{0}_m$, which implies that the state~$[\mb{x}^*;\mb{0}_{nm}]$, with $\mb{x}^* = \mb{1}_n \otimes \overline{ \mb{x}}^*$, is invariant under the proposed dynamics.
Following the Laplacian presentation of the consensus-type dynamics, Eqs.~\eqref{eq_xdot_g}-\eqref{eq_ydot_g} can be written as 
$\dot{\mb{x}} = \overline{W}_{\theta} q(\mb{x})-\alpha \mb{y}$ and 
 $\dot{\mb{y}} = \overline{W}_{\theta} q(\mb{x}) + H\dot{\mb{x}}
$ with $H:=\mbox{diag}[\nabla^2 f_i(\mb{x}_i)]$. Define time-dependent diogonal matrix $Q(t) :=  \mbox{diag}[\frac{q(\mb{x}_i)}{\mb{x}_i}]$ as a linearization matrix  to relate the linear state $\mb{x}$ with the quantized state $q(\mb{x})$ at every time instant $t$, i.e., $q(\mb{x}(t))=Q(t)\mb{x}(t)$. Given this matrix one can write $\overline{W}_\theta q(\mb{x}(t)) = \overline{W}_\theta Q(t)\mb{x}(t) = \overline{W}_{\theta,q} \mb{x}(t)$ where $\overline{W}_{\theta,q} := \overline{W}_{\theta}Q$. Then, in a compact form, we can rewrite the dynamics \eqref{eq_xdot_g}-\eqref{eq_ydot_g} in the following form:
\begin{align} \label{eq_xydot1}
	\left(\begin{array}{c} \dot{\mb{x}} \\ \dot{\mb{y}} \end{array} \right) = M_q(t,\alpha,\theta) \left(\begin{array}{c} {\mb{x}} \\ {\mb{y}} \end{array} \right),
\end{align}
where $M_q(t,\alpha,\theta)$ denotes the linearized matrix at time $t$ as
\begin{align} \label{eq_M_g}
	M_q(t,\alpha,\theta) = \left(\begin{array}{cc} \overline{W}_{\theta,q} \otimes I_m & -\alpha I_{mn} \\ H(\overline{W}_{\theta,q}\otimes I_m) & \overline{W}_{\theta,q} \otimes I_m - \alpha H
	\end{array} \right),
\end{align}

\vspace{0.2cm}
Most importantly, we can re-write~\eqref{eq_M_g} as follows.

\begin{lem} \label{lem_M}
	The proposed dynamics \eqref{eq_xydot1}-\eqref{eq_M_g} can be written as the following perturbation-based form,
	\begin{align}  \label{eq_Mg}
		M_q(t,\alpha,\theta) &=   M_q^0 + \alpha M^1, 
	\end{align}
	where  
	\begin{align} \label{eq_Mg2}
		M_q^0 = Q(t) M^0,
	\end{align} 
		with  $M^0 =   \left(\begin{array}{cc} \overline{W}_\theta \otimes I_m & \mb{0}_{mn\times mn} \\ H(\overline{W}_\theta \otimes I_m) & \overline{W}_\theta \otimes I_m \end{array} \right)$, and $M^1 = \left(\begin{array}{cc} \mb{0}_{mn\times mn} & - {I_{mn}} \\ {\mb{0}_{mn\times mn}} & - H \end{array} \right)$, where the following linear matrix inequalities hold: 
	\begin{align} \label{eq_Qlmi}
		(1-\frac{\rho}{2}) I_{nm}  &\preceq Q(t) \preceq (1+\frac{\rho}{2}) I_{nm}, \\ \label{eq_Mlmi}
		(1-\frac{\rho}{2}) M^0  &\preceq M_q^0 \preceq (1+\frac{\rho}{2}) M^0.
	\end{align}
\end{lem}
\begin{proof}
	Given~\eqref{eq_M_g}, collect the terms without $\alpha$ in $M_q^0$ and the rest of the terms in perturbation matrix $\alpha M^1$. From the definition of the log-scale quantization, we have 
	 $(1-\frac{\rho}{2}) \leq \frac{q(\mb{x}_i)}{\mb{x}_i} \leq (1+\frac{\rho}{2})$ and~\eqref{eq_Qlmi} readily follows. Let us, with some abuse of notation, denote $q(\mb{x}(t)) = Q(t) \mb{x}(t)$, and from the consensus algorithms, $-\sum_{j=1}^{n} w_{ij}^\theta (\mb{x}_i-\mb{x}_j)$ can be written as $(\overline{W}_{\theta} \otimes I_m)\mb{x}$. Then, following the \textit{\mbox{sector-bound} relation} of $q(\cdot)$, we have $-\sum_{j=1}^{n} w_{ij}^\theta (q(\mb{x}_i)-q(\mb{x}_j))=(\overline{W}_{\theta,q} \otimes I_m)\mb{x}$, where similar equations hold for $\mb{y}$. Then, $\overline{W}_{\theta,q} =  \overline{W}_{\theta} Q(t)$ and~\eqref{eq_Mg2} readily follows. Lastly, from the latter with~\eqref{eq_Qlmi}, we obtain~\eqref{eq_Mlmi}. 
\end{proof}

We now demonstrate that the dynamics presented in equations \eqref{eq_xdot_g}-\eqref{eq_ydot_g} provide a solution to \textbf{Problem 1}. Our approach involves proving that equations \eqref{eq_xydot1}-\eqref{eq_M_g} converge to the desired state. It is important to note that while the convergence analysis is conducted from a centralized perspective, the actual implementation remains fully distributed. The main theorem of this paper is presented as follows:

\begin{theorem} \label{thm_zeroeig}
	Under Assumptions~\ref{ass_cost}-\ref{ass_net}, $0<\rho<2$, and for sufficiently small $\alpha$, the proposed dynamics~\eqref{eq_xydot1}-\eqref{eq_M_g} converges to  state~$[\mb{x}^*;\mb{0}_{nm}]$, where  $\mb x^* := [\mb{x}^*_1;\mb{x}^*_2;\dots;\mb{x}^*_n] =  \mb{1}_n \otimes \overline{\mb{x}}^*$ is the solution to problem~\eqref{eq_prob}.
\end{theorem}
\begin{proof}
	
	\textbf{Step I:}
	From an eigenspectrum analysis point of view, one can find a relation between the spectrum $\sigma(\overline{W}_{\theta,q})$ and spectrum $\sigma(\overline{W}_{\theta})$. Recall that for $0<\rho<2$ from \eqref{eq_Qlmi} we have $Q(t)$ is a positive diagonal matrix, so we can write $\mbox{det}(\overline{W}_{\theta,q}-\lambda I_{2mn}) = \mbox{det}(\overline{W}_{\theta}-\lambda Q(t)^{-1} I_{2mn})=0$. In other words, given the eigenvalue $\lambda_i \in \sigma(\overline{W}_{\theta})$, we have $\lambda_i \xi_i \in \sigma(\overline{W}_{\theta,q})$ with $\xi_i := \frac{q(\mb{x}_i)}{\mb{x}_i}$. 
	Then, following the lower-block-triangular structure of $M^0$ and $M^0_q$, for the dynamics~\eqref{eq_xdot_g}-\eqref{eq_ydot_g} we have 
	\begin{align} \label{eq_spect_k}
		(1-\frac{\rho}{2}) \sigma(M^0) \leq \sigma(M^0_q) \leq (1+\frac{\rho}{2}) \sigma(M^0).
	\end{align}
This follows from Lemma~\ref{lem_M} and the fact that the matrix $M^0$ is block triangular and, therefore, $\sigma(M^0) = \sigma(\overline{W} \otimes I_m) \cup \sigma(\overline{W} \otimes I_m)$. From the definition of the Laplacian matrix  and its properties~\cite{SensNets:Olfati04}, $M^0$ has~$m$ set of eigenvalues as $\operatorname{Re}\{\lambda_{2n,j}\} \leq \ldots \leq \operatorname{Re}\{\lambda_{3,j}\} < \lambda_{2,j} = \lambda_{1,j} = 0,$
where $j=\{1,\ldots,m\}$. 
Then, from perturbation theory~\cite{stewart_book} and, in particular~\cite[Lemma~1]{cai2012average}, one can relate the eigenvalues of $\sigma(M^0)$, $\sigma(M_q^0)$, and $\sigma(M_q)$.
Considering the perturbation~$\alpha M^1$, we can check how the eigenvalues of \eqref{eq_spect_k} vary. 
In fact, we check how ~$\lambda_{1,j}$ and~$\lambda_{2,j}$ (the zero eigenvalues) move in the complex plane -- the zero eigenvalues move to the left-half plane the system is stable and unstable if they move toward the right-half plane.  Let~$\lambda_{1,j}(\alpha,t)$ and~$\lambda_{2,j}(\alpha,t)$ denote  the perturbed eigenvalues by~$\alpha M^1$ and follows \eqref{eq_spect_k}. Invoking Assumption~\ref{ass_net}, the following properties hold for  Laplacian $\overline{W}_\theta$ of a strongly connected weight-balanced graph~\cite{SensNets:Olfati04}: \emph{(i)} $\overline{W}_\theta$ has only one zero eigenvalue and the other eigenvalues on the left-half plane, and \emph{(i)} we have~$\mb{1}_n^\top \overline{W}_\theta= \mb{0}_n$ and~$\overline{W}_\theta \mb{1}_n=\mb{0}_n$, i.e., $\mb{1}_n^\top$ and~$\mb{1}_n$ are respectively the left and right eigenvector of $\overline{W}_\theta$ for the zero eigenvalue. Hence,  the right eigenvectors of~$\lambda_{1,j}$ and~$\lambda_{2,j}$ are as follows:
	\begin{align} \nonumber
		V = [V_1~V_2] =\left(\begin{array}{cc}
			\mb{1}_n& \mb{0}_n \\
			\mb{0}_n & \mb{1}_n
		\end{array} \right)\otimes I_m,
	\end{align}
with the left eigenvectors as~$V^\top$.

From \eqref{eq_M_g}, we have~$\frac{dM_q(\alpha)}{d\alpha}|_{\alpha=0}=M^1$ and, subsequently,
	
	\small
	\begin{eqnarray} \label{eq_dmalpha}
		V^\top M^1 V= \left(\begin{array}{cc}
			\mb{0}_{m\times m}	& \mb{0}_{m\times m} \\
			...	& -(\mb{1}_n \otimes I_m)^\top H (\mb{1}_n \otimes I_m)
		\end{array} \right),
	\end{eqnarray} \normalsize
which has $m$ zero eigenvalues.
To determine the other  $m$ eigenvalues, we recall that, by invoking Assumption~\ref{ass_cost},  $f_i(\cdot)$ satisfies
	\begin{equation} 
		-(\mb{1}_n \otimes I_m)^\top H  ( \mb{1}_n \otimes I_m)= -\sum_{i=1}^n \nabla^2  f_i(\mb{x}_i) \prec 0.
	\end{equation}
Therefore, following from the perturbation-based analysis in \cite[Lemma~1]{cai2012average}, ${\frac{d\lambda_{1,j}}{d\alpha}|_{\alpha=0} = 0}$ and~${\frac{d\lambda_{2,j}}{d\alpha}}|_{\alpha=0}<0$. This implies that perturbing the matrix $M^0_q$ by~$\alpha M^1$, its~$m$ zero eigenvalues~$\lambda_{2,j}(\alpha,t)$ for all $j \in \{1,\dots,m\}$ change toward stability (i.e. left-half plane) while  $\lambda_{1,j}(\alpha,t)$ remain zero.

\textbf{Step II:}
Next, we need to ensure that due to the perturbation $\alpha M^1$ no other eigenvalue of $M_q^0$ moves to the right-half plane. For this purpose, we first recall some eigenspectrum analysis
from \cite[Appendix]{delay_est}. For notation simplicity, we drop the dependence on $\theta, t$. We find $\sigma(M_q)$ via row/column permutations in \cite[Eq.~(18)]{delay_est} as follows:
 \begin{align} \nonumber
	\mbox{det}(&\alpha  I_{mn}) \mbox{det}\Big(H(\overline{W}\otimes I_m) +\\ &\qquad \qquad  \frac{1}{\alpha}(\overline{W} \otimes I_m - \alpha H -\lambda I_{mn})^2\Big) = 0.
\end{align}
Without loss of generality, to simplify the calculations, we set $m=1$, and we have
\begin{align} \label{eq_m=1}
	\mbox{det}(I_{n}) \mbox{det}((\overline{W}  - \lambda I_{n})^2 +\alpha \lambda H ) = 0.
\end{align}

To preserve stability, the admissible $\alpha$ range needs to be defined such that the eigenvalues are in the left-half plane, except for one zero eigenvalue. Furthermore, note that the eigenvalues of a system matrix are continuous functions of its elements~\cite{stewart_book}. Additionally,
notice that~\eqref{eq_m=1} has two roots; one root is $\alpha=0$;  let the other positive root be denoted by $\overline{\alpha}>0$. Then, the range of $\alpha \in (0,~\overline{\alpha})$ is admissible for stability of $M_q(\alpha)$. We need to find $\overline{\alpha}$ as the positive root of~\eqref{eq_m=1}.
Note that for $\alpha = 0$, we have $\mbox{det}((\overline{W}  - \lambda I_{n})^2) = 0$, and the eigenspectrum is equal to two sets of eigenvalues in~$\sigma(\overline{W})$.



The diagonal structure of $\alpha H$ allows us to rewrite~\eqref{eq_m=1} (with some abuse of notation) in the following form:
\begin{align} \nonumber
	\mbox{det}\left((\overline{W}  -\lambda I_{n} + \sqrt{\alpha |\lambda| H})(\overline{W}  - \lambda I_{n} - \sqrt{\alpha |\lambda| H} )\right)  = 0,
\end{align}
with $\lambda$ as the eigenvalue of $\overline{W}$, which can be rewritten as
\begin{align} \nonumber
	\mbox{det}(\overline{W}  &- \lambda I_{n} \pm \sqrt{\alpha |\lambda| H} )= \\ \label{eq_det_all} &\qquad 
 \qquad \mbox{det}\left(\overline{W}  -\lambda \left(1   \pm \sqrt{\frac{\alpha  H}{|\lambda|}} \right)I_{n}\right) = 0.
\end{align}
Therefore, the perturbed eigenvalue of $\overline{W}$ from~\eqref{eq_det_all} is $\lambda (1 \pm \sqrt{\frac{\alpha  H}{|\lambda|}})$. Then, we find $\overline{\alpha}$ as the min value of $\alpha$ to make this perturbation zero;
\begin{align}
	\overline{\alpha} = \argmin_{\alpha} \left|1 - \sqrt{\frac{\alpha  H}{|\lambda|}}\right|.
\end{align}
Recall from Assumption~\ref{ass_cost} that $H \preceq \gamma I_{n}$, so we have
\begin{align}
\argmin_{\alpha} |1 - \sqrt{\frac{\alpha  H}{|\lambda|}}|
\geq \frac{\min \{|\operatorname{Re}\{\lambda_j\}|\neq 0\}}{\max \{H_{ii}\}} = \frac{|\operatorname{Re}\{\lambda_2\}|}{\gamma}.
\end{align}
Then, the stability range of $\alpha$ is
\begin{align} \label{eq_alphabar0}
	0 < \alpha < \overline{\alpha}:= \frac{\min_{1\leq j\leq m}|\operatorname{Re}\{{\lambda}_{2,j}\}|}{\gamma}.
\end{align}
For these values of $\alpha$ the system matrix $M_q(\alpha)$ has only one set of $m$ zero eigenvalue (associated with the dimension $m$ of state vector $\mb{x}_i$) and other eigenvalues in the left-half plane.

\textbf{Step III:}
Define the residual on the state and auxiliary variable as
$\delta  = \left(\begin{array}{c} {\mb{x}} \\ {\mb{y}} \end{array} \right) - \left(\begin{array}{c} {\mb{x}}^* \\ \mb{0}_{nm} \end{array} \right) \in \mathbb{R}^{2nm}.$
Recall that $[\mb{x}^*;\mb{0}_{nm}]$ is invariant under the proposed dynamics~\eqref{eq_xydot1}-\eqref{eq_M_g}. Therefore,
\begin{align} \label{eq_ddelta}
	\dot{\delta} &=  M_q \left(\begin{array}{c} \mb{x} \\ \mb{y} \end{array} \right) - M_q\left(\begin{array}{c} \mb{x}^* \\ \mb{0}_{nm} \end{array} \right) = M_q \delta.
\end{align}
Consider the following positive definite function  ${U(\delta) = \frac{1}{2} \delta^\top \delta}$ as a Lyapunov function candidate. From~\eqref{eq_ddelta}, $\dot{U} = {\delta}^\top \dot{\delta}=  \delta^\top M_q {\delta}$. Following Step I and Step~II, $\lambda_{2,j}$ as the largest nonzero eigenvalue of $M_q$ is in the left-half plane. Then, from~\cite[Sections~VIII-IX]{SensNets:Olfati04}, we have
\begin{eqnarray} \label{eq_Re2}
	\dot{U} \leq \max_{1\leq j\leq m}\operatorname{Re}\{{\lambda}_{2,j}\} \delta^\top  \delta,
\end{eqnarray}
which is negative-definite for $\delta \neq \mb{0}_{2nm}$. Hence, it is a Lyapunov function, and by invoking Lyapunov's theorem, we have that $\delta \rightarrow \mb{0}_{2nm}$.
\end{proof}
Theorem~\ref{thm_zeroeig} implies that the underlying system is asymptotically globally stable for $0<\alpha<\overline{\alpha}$ with the equilibrium $[\mb{1}_n \otimes \overline{ \mb{x}}^*;\mb{0}_{nm}]$ being invariant under the proposed dynamics.
The above proves that the solution converges for sufficiently small $\alpha$ given by~\eqref{eq_alphabar0}. However, the exact convergence rate cannot be obtained due to the nonlinear/nonsmooth nature of the proposed dynamics. For the linearized case given by Eqs. \eqref{eq_xydot1}-\eqref{eq_M_g}, following the decay rate of the dynamical systems and Eq.~\eqref{eq_Re2}, the rate of convergence is (at least) $\mc{O}(\exp(\max_{1\leq j\leq m}\operatorname{Re}\{{\lambda}_{2,j}\}t))$. 
Moreover, it can be seen that the convergence is irrespective of the switching signal topology $\theta$ for any $\alpha$ satisfying Eq.~\eqref{eq_alphabar0}. In other words, to ensure convergence, we consider the minimum of $|\operatorname{Re}\{\lambda_{2,j}\}|$ in~\eqref{eq_alphabar0} for the set of all network topologies $\mc{G}_\theta$.
Specifically, this follows as a consequence of the fact that $M_q(t,\alpha,\theta)$ has only one set of $m$ zero eigenvalues with the rest of eigenvalues in the left-half plane irrespective of the switching signal $\theta$ for any $\alpha \in (0~\overline{\alpha})$ -- see details in the proof of Theorem~\ref{thm_zeroeig}. 

\section{Simulation: Application to SVM} \label{sec_sim}

Consider data points~${\boldsymbol{\chi}_i \in \mathbb{R}^{m-1}}$, ${i=1,\ldots,N}$ labeled by~${l_i \in \{-1,1\}}$ (binary classification). These points are distributed among a network of nodes/agents and every agent has its local objective function. The agents share information over resource-constrained communication networks (e.g. wireless networks where each agent has constraints on the bandwidth). Therefore, the communication channels are (logarithmically) quantized. Over this network, the goal is to find the global optimal max-margin SVM classifier (hyperplane) ${\boldsymbol{\omega}^\top \boldsymbol{\chi} - \nu =0}$  to classify the data points into two groups in $\mathbb{R}^{m-1}$. The parameters associated to this hyperplane are $\boldsymbol{\omega}~$ and~$\nu$ that are a solution to the following optimization problem~\cite{chapelle2007training}:
\begin{align} \label{eq_svm_cent}
	\begin{aligned}
		\displaystyle
		& \min_{\boldsymbol{\omega},\nu}
		~ &  \boldsymbol{\omega}^\top \boldsymbol{\omega} + C \sum_{j=1}^{N} \max\{1-l_j( \boldsymbol{\omega}^\top \boldsymbol{\chi}_j-\nu),0\}^p,
	\end{aligned}
\end{align}
with~${p \in \mathbb N}$ and $C>0$ as the margin size parameter. Replace non-differentiable $\max\{z,0\}^1$ with its twice differentiable equivalent~${L(z,\mu)=\frac{1}{\mu}\log (1+\exp(\mu z))}$ to get the hinge loss function \cite{slp_book}.
One can make the two arbitrarily close by selecting $\mu$ large enough~\cite{slp_book} (for simulation we set $\mu=2$); thus, the distributed setup is formulated as follows:
\begin{align} \label{eq_svm_dist}
	\begin{aligned}
		\displaystyle
		\min_{\boldsymbol{\omega}_1,\nu_1,\ldots,\boldsymbol{\omega}_n,\nu_n}
		\quad &  \sum_{i=1}^{n} \boldsymbol{\omega}_i^\top \boldsymbol{\omega}_i + C \sum_{j=1}^{N_i} \tfrac{1}{\mu}\log (1+\exp(\mu z_{i,j})) \\
		\text{subject to} \quad&  \boldsymbol{\omega}_1 = \dots = \boldsymbol{\omega}_n,\qquad\nu_1 = \dots =\nu_n,
	\end{aligned}
\end{align}
where $N_i$ is the number of data points at node $i$, ${z_{i,j}:=1-l_j( \boldsymbol{\omega}_i^\top \boldsymbol{\chi}^i_j-\nu_i)}$ with the set of data points $\boldsymbol{\chi}^i_j$ assigned to agent $i$. These data are distributed among~$n$ agents where they reach consensus on the local hyperplane classifier parameters $\boldsymbol{\omega}_i$ and~$\nu_i$.

We consider a WB network of $n=20$ agents to classify the data in 2D space, which implies that the hyper-surface (line) to classify the data can be represented with $m=3$ variables, i.e., $\boldsymbol{\omega}_i \in \mathbb R^2 $ and~$\nu_i  \in \mathbb R$. The communication network of agents $\mc{G}_\theta$ is considered an \mbox{Erd\H{o}s-Rényi} network with linking probability $30\%$. This network topology $\mc{G}_\theta$ changes every $t=0.1$ sec. We consider $m=100$ data points to be classified via optimization \eqref{eq_svm_dist}. Every agent has access to $75\%$ of randomly chosen data points (with possible overlaps) and updates its states via the local dynamics \eqref{eq_xdot_g}-\eqref{eq_ydot_g}. The parameters for the simulation are as follows: $\alpha = 0.1$, $\mu = 2$, $C=40$, and $\rho = 0.25$.
The time-evolution of the classifier parameters at agents is shown in Fig.~\ref{fig_param}, where the consensus is achieved among all the agents. The hyperplane parameters reach a consensus as $\boldsymbol{\omega}_i = [1,-1] $ and~$\nu_i =0$. 
\begin{figure} [bpt]
	\centering
	\includegraphics[width=1.71in]{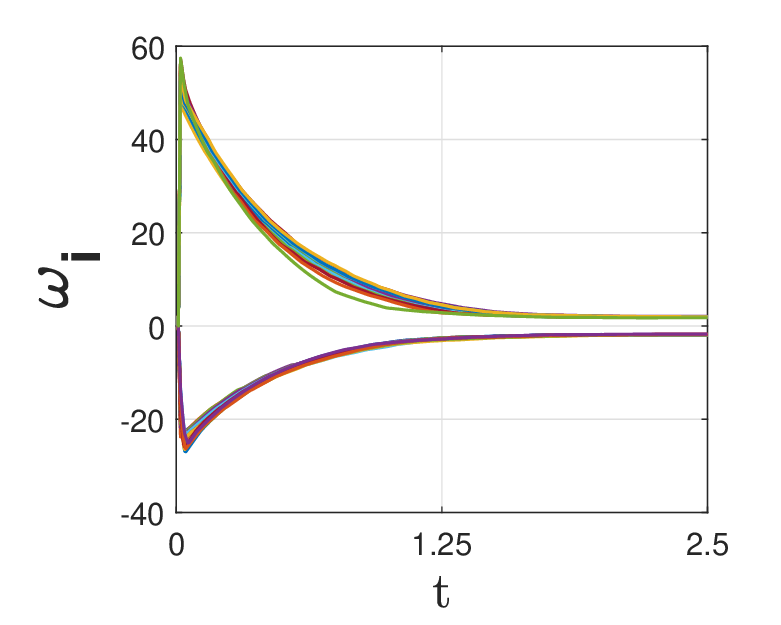}
	\includegraphics[width=1.71in]{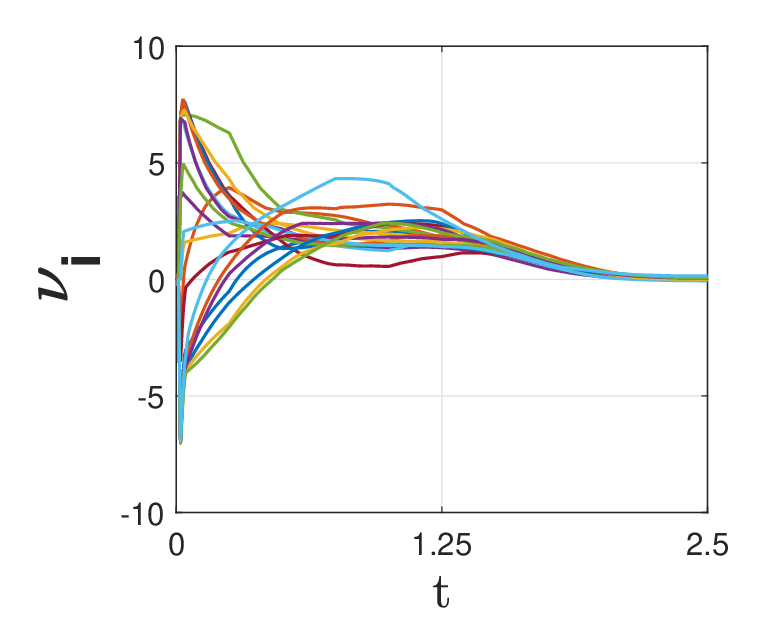}
	\caption{The time-evolution of the classifying line parameters $\boldsymbol{\omega}_i \in \mathbb R^2 $ and~$\nu_i  \in \mathbb R$ are shown in this figure. As it is clear from the figure, all agents have reached a consensus over these parameters. The minor oscillations in $\nu_i$ states are due to changes in the network topology.} \label{fig_param}
\end{figure}
The SVM lines by different agents separating the data points are shown at two different time instants in Fig.~\ref{fig_line}. 

\begin{figure} [bpt]
	\centering
	\includegraphics[width=1.5in]{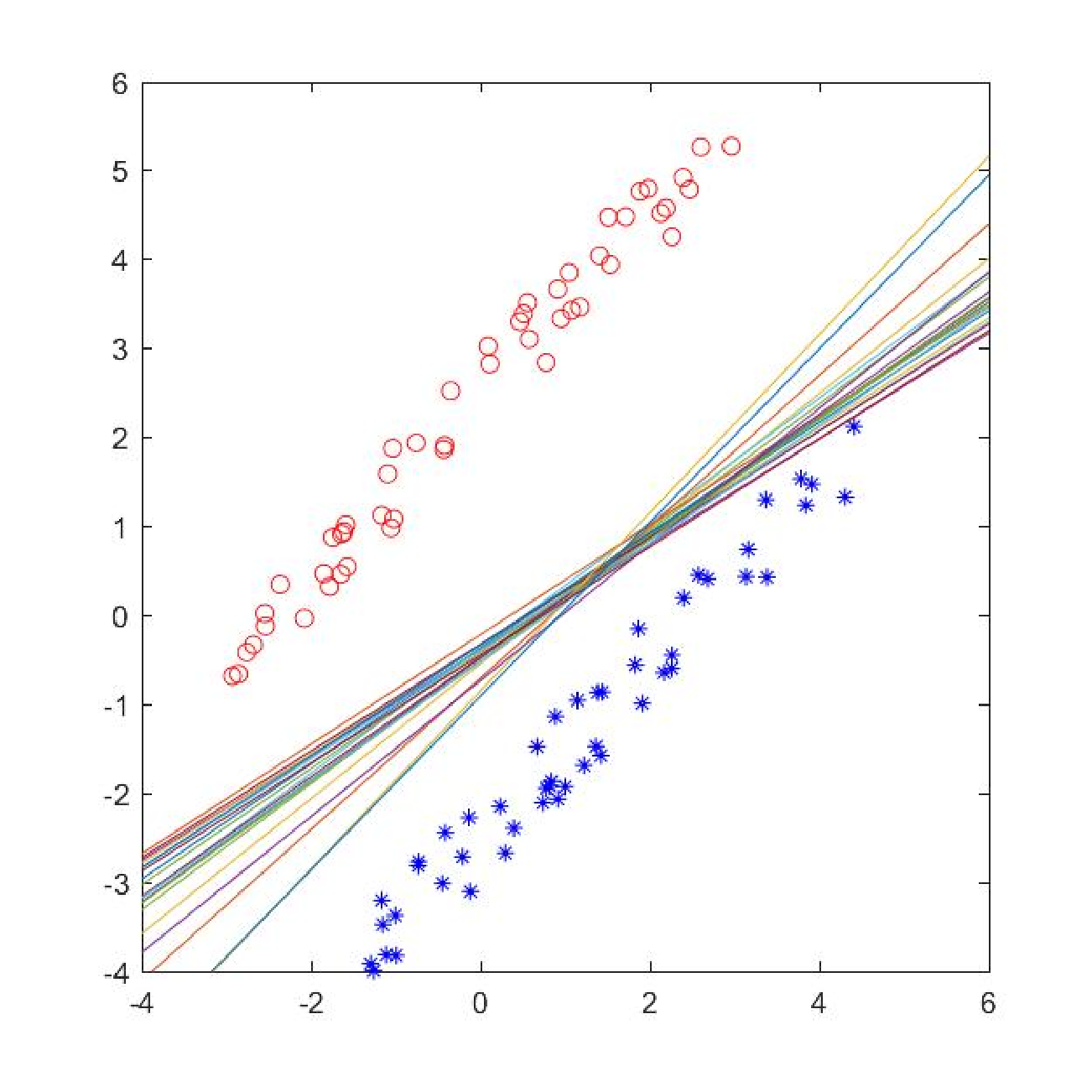}
	\includegraphics[width=1.5in]{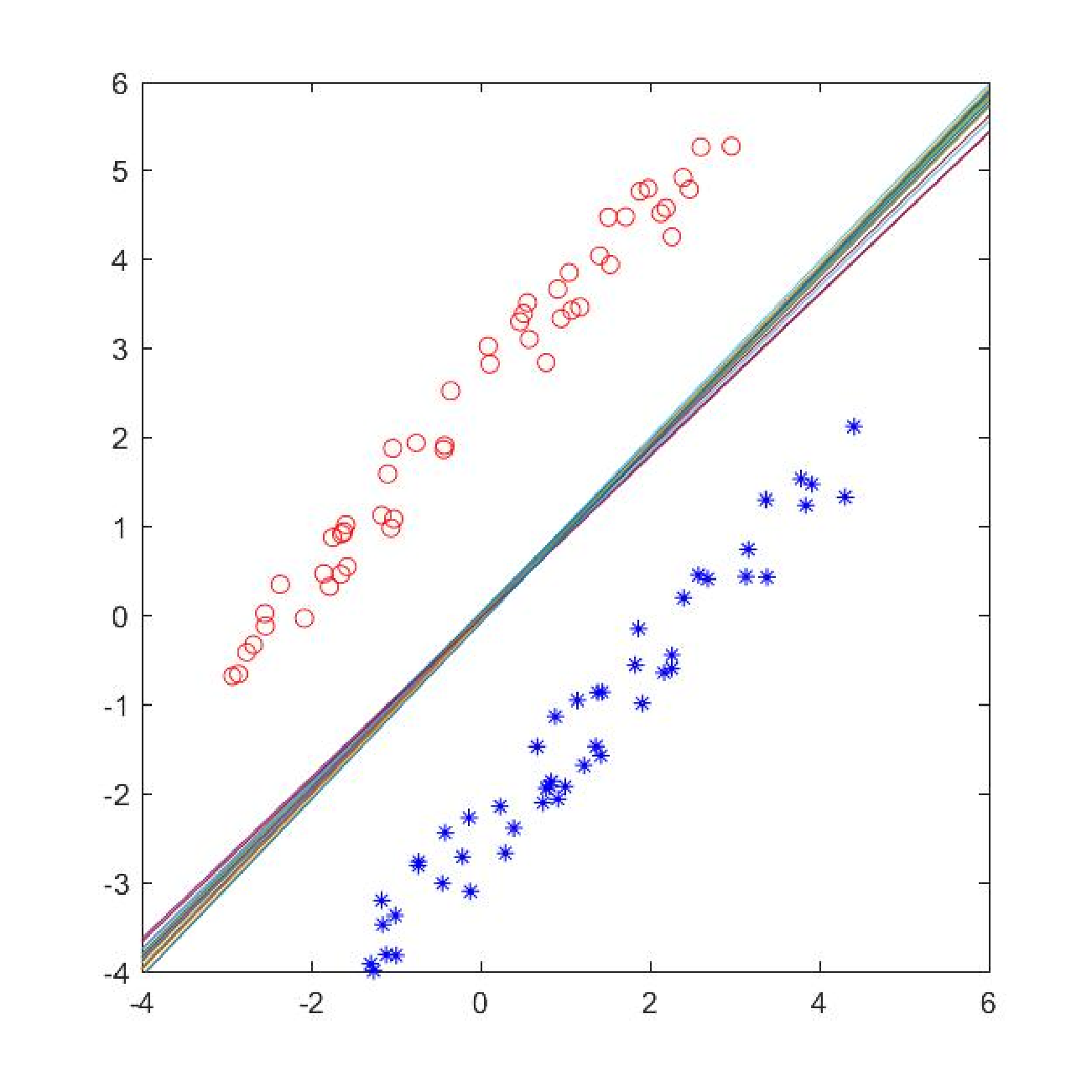}
	\caption{The separating SVM lines by different agents at time $t=1.25$~sec (Left) and time $t=2.5$ sec (Right).} \label{fig_line}
\end{figure}

\textbf{Comparison with the literature:} We compare our quantized solution ($\rho = 0.125$) with some existing linear \cite{dsvm} and nonlinear distributed optimization solutions; namely, the ones by Li \emph{et. al.}~\cite{li2020time} and Ning \emph{et. al.} \cite{ning2017distributed}. These use nonlinear sign-based protocols to improve the convergence rate. The multi-agent network is considered an undirected cycle. The simulation is given in Fig.~\ref{fig_compare}-Left. Note that the advantage of our proposed log-quantized protocol is that it can handle log-scale quantization with no optimality gap while the other three solutions cannot handle quantization. Clearly from the figure, our algorithm converges toward an optimal solution in the \textit{presence of quantization} (the residual decays toward zero) with similar performance as compared to the other \textit{ideal methods with no quantization}. 
\begin{figure} [hbpt]
	\centering
	\includegraphics[width=1.71in]{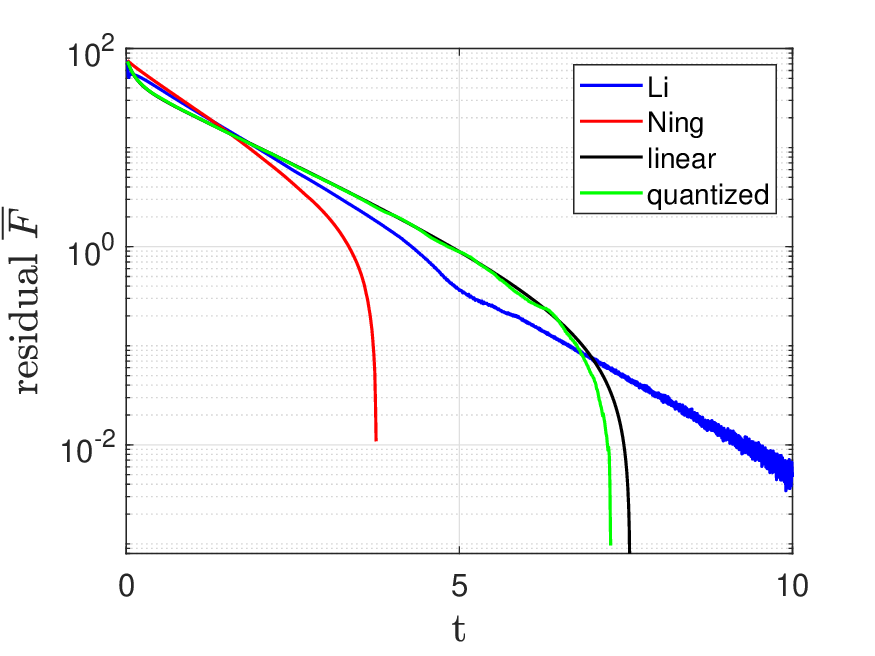}
		\includegraphics[width=1.71in]{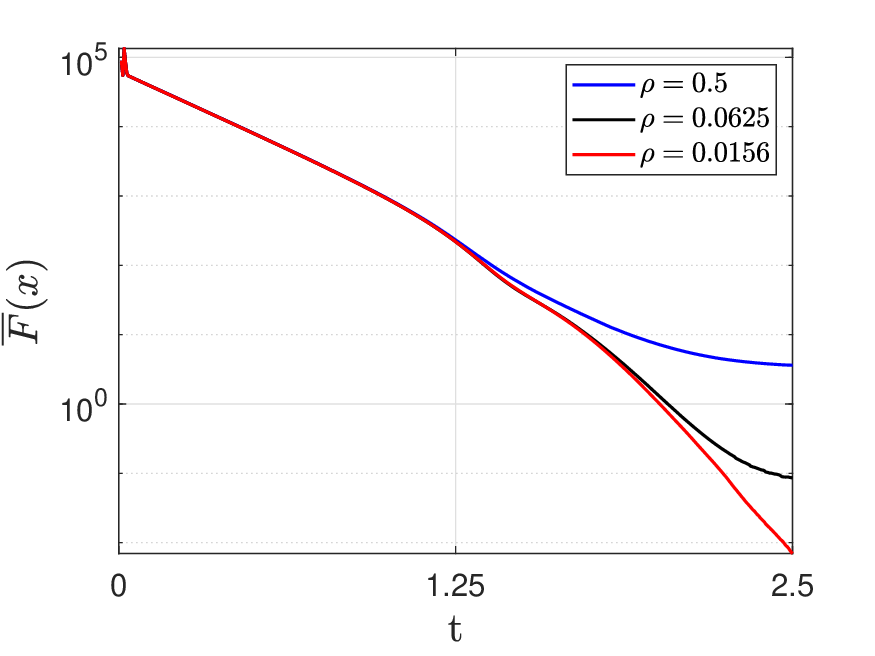}
	\caption{(Left) Comparison between different distributed optimization solutions with our proposed quantized setup, (Right) Comparison of the residual convergence for different logarithmic quantization levels $\rho$.} \label{fig_compare}
\end{figure}
For the next simulation, in Fig.~\ref{fig_compare}-Right we compare the convergence rate of the residual function (optimality gap) $\overline{F}(\mb{x})=F(\mb{x})-F(\mb{x}^*)$ for different log quantization levels $\rho$ in~\eqref{eq_quan_log}. The residual value decreases over time, and its decay rate increases for smaller values of $\rho$ (i.e., finer quantization level). In other words,  smaller $\rho$ results in faster convergence.

\textbf{Comparison with uniform quantizer:} For the last simulation, we compare the residual cost $\overline{F}(\mb{x})$ and convergence under uniform quantization versus log-scale quantization. Recall that uniform quantization is defined as
$q_u(z) = \rho\left[\dfrac{z}{\rho}\right],
$
where $\rho$ is the uniform quantization level and $[\cdot]$ denotes rounding to the nearest integer.
For the simulation, we set the quantization level (for both uniform and log-scale) equal to $\rho = 0.125$ and repeat the SVM hinge loss optimization under both quantization scenarios. The residual costs are compared in Fig.~\ref{fig_uni}. As it is clear from the figure, the solution under uniform quantization results in a larger steady-state residual than the logarithmic quantization. This is well-established in the literature that uniform quantization results in some optimality gap. This is because of the constant quantization error in the uniform case which adds gradient error near the optimal point, while in the log-scale quantization, we have finer quantization near the optimal point and less gradient error that converges to zero. 
\begin{figure} [hbpt!]
	\centering
	\includegraphics[width=1.7in]{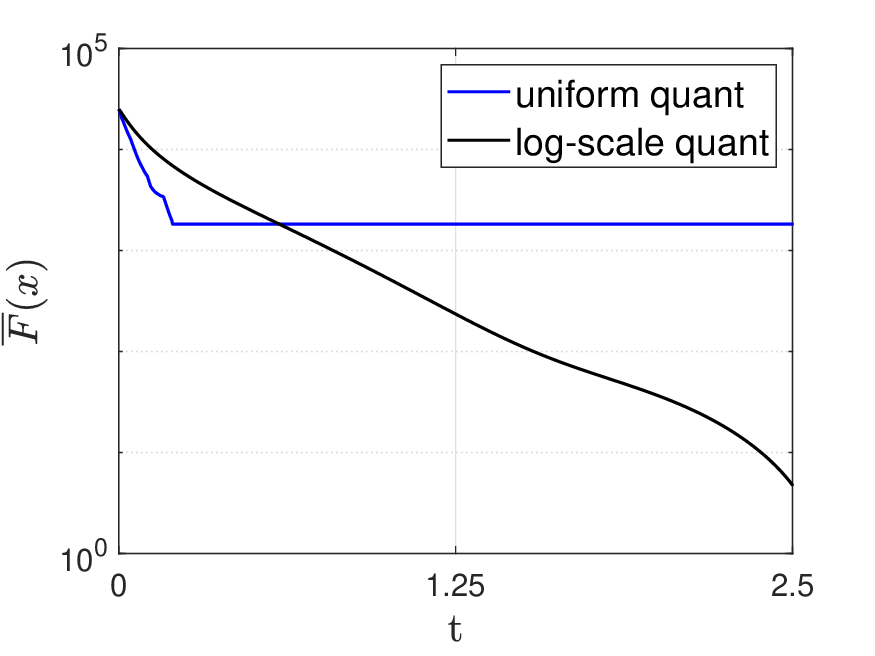}
	\caption{Residual under uniform vs. log-scale quantization: Clearly, optimization under uniform quantization results in large optimality gap.} \label{fig_uni}
\end{figure}

\section{Conclusion and Future Research}\label{sec_con}
We have presented a novel approach to distributed optimization over dynamic multi-agent networks using logarithmic quantization. Our method addresses the critical challenge of limited bandwidth in real-world networks while maintaining high precision in optimization tasks. 
Finding the convergence rate for general costs satisfying Assumption~\ref{ass_cost} and under log-quantized nonlinearity is our direction of future research.

\section*{Acknowledgement}
This work is funded by Semnan University, research grant No. 226/1403/1403214.

\bibliographystyle{IEEEbib}
\bibliography{bibliography}

\begin{thebibliography}{10}

\bibitem{ren2010distributed}
W.~Ren and Y.~Cao,
\newblock {\em Distributed coordination of multi-agent networks: emergent
  problems, models, and issues},
\newblock Springer Science \& Business Media, 2010.

\bibitem{terelius2011decentralized}
H.~Terelius, U.~Topcu, and R.~M. Murray,
\newblock ``Decentralized multi-agent optimization via dual decomposition,''
\newblock {\em IFAC proceedings volumes}, vol. 44, no. 1, pp. 11245--11251,
  2011.

\bibitem{zayyani2018bayesian}
H.~Zayyani, M.~Korki, and F.~Marvasti,
\newblock ``Bayesian hypothesis testing detector for one bit diffusion lms with
  blind missing samples,''
\newblock {\em Signal Processing}, vol. 146, pp. 61--65, 2018.

\bibitem{pan2024utilizing}
Z.~Pan, H.~Yang, and H.~Liu,
\newblock ``Utilizing second-order information in noisy information-sharing
  environments for distributed optimization,''
\newblock in {\em IEEE International Conference on Acoustics, Speech and Signal
  Processing}, 2024, pp. 9156--9160.

\bibitem{wang2022gradient}
Y.~Wang and T.~Ba{\c{s}}ar,
\newblock ``Gradient-tracking-based distributed optimization with guaranteed
  optimality under noisy information sharing,''
\newblock {\em IEEE Transactions on Automatic Control}, vol. 68, no. 8, pp.
  4796--4811, 2022.

\bibitem{reisizadeh2022distributed}
H.~Reisizadeh, B.~Touri, and S.~Mohajer,
\newblock ``Distributed optimization over time-varying graphs with imperfect
  sharing of information,''
\newblock {\em IEEE Transactions on Automatic Control}, vol. 68, no. 7, pp.
  4420--4427, 2022.

\bibitem{zhang2018distributed}
H.~Zhang, F.~Teng, Q.~Sun, and Q.~Shan,
\newblock ``Distributed optimization based on a multiagent system disturbed by
  general noise,''
\newblock {\em IEEE Transactions on Cybernetics}, vol. 49, no. 8, pp.
  3209--3213, 2018.

\bibitem{qu2017harnessing}
G.~Qu and N.~Li,
\newblock ``Harnessing smoothness to accelerate distributed optimization,''
\newblock {\em IEEE Transactions on Control of Network Systems}, vol. 5, no. 3,
  pp. 1245--1260, 2017.

\bibitem{spl24}
M.~Doostmohammadian and A.~Aghasi,
\newblock ``Accelerated distributed allocation,''
\newblock {\em IEEE Signal Processing Letters}, 2024.

\bibitem{shi2018augmented}
C.~Shi and G.~Yang,
\newblock ``Augmented lagrange algorithms for distributed optimization over
  multi-agent networks via edge-based method,''
\newblock {\em Automatica}, vol. 94, pp. 55--62, 2018.

\bibitem{yang2019survey}
T.~Yang, X.~Yi, J.~Wu, Y.~Yuan, Di. Wu, Z.~Meng, Y.~Hong, H.~Wang, Z.~Lin, and
  K.~H. Johansson,
\newblock ``A survey of distributed optimization,''
\newblock {\em Annual Reviews in Control}, vol. 47, pp. 278--305, 2019.

\bibitem{nedic2018distributed}
A.~Nedi{\'c} and J.~Liu,
\newblock ``Distributed optimization for control,''
\newblock {\em Annual Review of Control, Robotics, and Autonomous Systems},
  vol. 1, pp. 77--103, 2018.

\bibitem{dimakis2010gossip}
A.~G. Dimakis, S.~Kar, J.~M.~F. Moura, M.~G. Rabbat, and A.~Scaglione,
\newblock ``Gossip algorithms for distributed signal processing,''
\newblock {\em Proceedings of the IEEE}, vol. 98, no. 11, pp. 1847--1864, 2010.

\bibitem{lee2017lognet}
E.~H. Lee, D.~Miyashita, E.~Chai, B.~Murmann, and S.~S. Wong,
\newblock ``Lognet: Energy-efficient neural networks using logarithmic
  computation,''
\newblock in {\em IEEE International Conference on Acoustics, Speech and Signal
  Processing (ICASSP)}. IEEE, 2017, pp. 5900--5904.

\bibitem{bu2017stability}
X.~Bu, Z.~Hou, L.~Cui, and J.~Yang,
\newblock ``Stability analysis of quantized iterative learning control systems
  using lifting representation,''
\newblock {\em International Journal of Adaptive Control and Signal
  Processing}, vol. 31, no. 9, pp. 1327--1336, 2017.

\bibitem{oh2021automated}
S.~Oh, H.~Sim, S.~Lee, and J.~Lee,
\newblock ``Automated log-scale quantization for low-cost deep neural
  networks,''
\newblock in {\em IEEE/CVF Conference on Computer Vision and Pattern
  Recognition}, 2021, pp. 742--751.

\bibitem{gholami2022survey}
A.~Gholami, S.~Kim, Z.~Dong, Z.~Yao, M.~W. Mahoney, and K.~Keutzer,
\newblock ``A survey of quantization methods for efficient neural network
  inference,''
\newblock in {\em Low-Power Computer Vision}, pp. 291--326. Chapman and
  Hall/CRC, 2022.

\bibitem{miyashita2016convolutional}
D.~Miyashita, E.~H. Lee, and B.~Murmann,
\newblock ``Convolutional neural networks using logarithmic data
  representation,''
\newblock {\em arXiv preprint arXiv:1603.01025}, 2016.

\bibitem{jiang2024hardware}
T.~Jiang, L.~Xing, J.~Yu, and J.~Qian,
\newblock ``A hardware-friendly logarithmic quantization method for {CNNs} and
  {FPGA} implementation,''
\newblock {\em Journal of Real-Time Image Processing}, vol. 21, no. 4, pp. 108,
  2024.

\bibitem{ojsys}
M.~Doostmohammadian, A.~Aghasi, A.~I. Rikos, A.~Grammenos, E.~Kalyvianaki,
  C.~N. Hadjicostis, K.~H. Johansson, and T.~Charalambous,
\newblock ``Distributed anytime-feasible resource allocation subject to
  heterogeneous time-varying delays,''
\newblock {\em IEEE Open Journal of Control Systems}, vol. 1, pp. 255--267,
  2022.

\bibitem{my_ecc}
M.~Doostmohammadian, A.~Aghasi, A.~I. Rikos, A.~Grammenos, E.~Kalyvianaki,
  C.~N. Hadjicostis, K.~H. Johansson, and T.~Charalambous,
\newblock ``Fast-convergent anytime-feasible dynamics for distributed
  allocation of resources over switching sparse networks with quantized
  communication links,''
\newblock in {\em European Control Conference}. IEEE, 2022, pp. 84--89.

\bibitem{ccta}
M.~Doostmohammadian, A.~Aghasi, M.~Vrakopoulou, and T.~Charalambous,
\newblock ``1st-order dynamics on nonlinear agents for resource allocation over
  uniformly-connected networks,''
\newblock in {\em IEEE Conference on Control Technology and Applications}.
  IEEE, 2022, pp. 1184--1189.

\bibitem{rabbat2005quantized}
M.~G. Rabbat and R.~D. Nowak,
\newblock ``Quantized incremental algorithms for distributed optimization,''
\newblock {\em IEEE Journal on Selected Areas in Communications}, vol. 23, no.
  4, pp. 798--808, 2005.

\bibitem{lee2018finite}
C.~Lee, N.~Michelusi, and G.~Scutari,
\newblock ``Finite rate quantized distributed optimization with geometric
  convergence,''
\newblock in {\em 52nd Asilomar Conference on Signals, Systems, and Computers}.
  IEEE, 2018, pp. 1876--1880.

\bibitem{lee2019deep}
H.~Lee, S.~Lee, and T.~Quek,
\newblock ``Deep learning for distributed optimization: Applications to
  wireless resource management,''
\newblock {\em IEEE Journal on Selected Areas in Communications}, vol. 37, no.
  10, pp. 2251--2266, 2019.

\bibitem{rikos2023distributed}
A.~I. Rikos, W.~Jiang, T.~Charalambous, and K.~H. Johansson,
\newblock ``Distributed optimization with gradient descent and quantized
  communication,''
\newblock in {\em Proceedings of 22nd IFAC World Congress}, 2023, pp.
  6433--6439.

\bibitem{rikos2021optimal}
A.~I. Rikos, A.~Grammenos, E.~Kalyvianaki, C.~N. Hadjicostis, T.~Charalambous,
  and K.~H. Johansson,
\newblock ``Optimal {CPU} scheduling in data centers via a finite-time
  distributed quantized coordination mechanism,''
\newblock in {\em 60th IEEE Conference on Decision and Control (CDC)}. IEEE,
  2021, pp. 6276--6281.

\bibitem{xin2020decentralized}
R.~Xin, S.~Kar, and U.~A. Khan,
\newblock ``Decentralized stochastic optimization and machine learning: A
  unified variance-reduction framework for robust performance and fast
  convergence,''
\newblock {\em IEEE Signal Processing Magazine}, vol. 37, no. 3, pp. 102--113,
  2020.

\bibitem{qureshi2021decentralized}
M.~I. Qureshi, R.~Xin, S.~Kar, and U.~A. Khan,
\newblock ``A decentralized variance-reduced method for stochastic optimization
  over directed graphs,''
\newblock in {\em IEEE International Conference on Acoustics, Speech and Signal
  Processing (ICASSP)}. IEEE, 2021, pp. 5030--5034.

\bibitem{dsvm}
M.~Doostmohammadian, A.~Aghasi, T.~Charalambous, and U.~A. Khan,
\newblock ``Distributed support vector machines over dynamic balanced directed
  networks,''
\newblock {\em IEEE Control Systems Letters}, vol. 6, pp. 758--763, 2021.

\bibitem{qureshi2020s}
M.~I. Qureshi, R.~Xin, S.~Kar, and U.~A. Khan,
\newblock ``{S-ADDOPT}: Decentralized stochastic first-order optimization over
  directed graphs,''
\newblock {\em IEEE Control Systems Letters}, vol. 5, no. 3, pp. 953--958,
  2020.

\bibitem{nowzari2016distributed}
C.~Nowzari and J.~Cort{\'e}s,
\newblock ``Distributed event-triggered coordination for average consensus on
  weight-balanced digraphs,''
\newblock {\em Automatica}, vol. 68, pp. 237--244, 2016.

\bibitem{guo2013consensus}
M.~Guo and D.~V. Dimarogonas,
\newblock ``Consensus with quantized relative state measurements,''
\newblock {\em Automatica}, vol. 49, no. 8, pp. 2531--2537, 2013.

\bibitem{SensNets:Olfati04}
R.~Olfati-Saber and R.~M. Murray,
\newblock ``Consensus problems in networks of agents with switching topology
  and time-delays,''
\newblock {\em IEEE Transactions on Automatic Control}, vol. 49, no. 9, pp.
  1520--1533, Sept. 2004.

\bibitem{stewart_book}
G.~W. Stewart and J.~Sun,
\newblock {\em Matrix perturbation theory},
\newblock Academic Press, 1990.

\bibitem{cai2012average}
K.~Cai and H.~Ishii,
\newblock ``Average consensus on general strongly connected digraphs,''
\newblock {\em Automatica}, vol. 48, no. 11, pp. 2750--2761, 2012.

\bibitem{delay_est}
M.~Doostmohammadian, M.~Pirani, U.~A. Khan, and T.~Charalambous,
\newblock ``Consensus-based distributed estimation in the presence of
  heterogeneous, time-invariant delays,''
\newblock {\em IEEE Control Systems Letters}, vol. 6, pp. 1598 -- 1603, 2021.

\bibitem{chapelle2007training}
O.~Chapelle,
\newblock ``Training a support vector machine in the primal,''
\newblock {\em Neural computation}, vol. 19, no. 5, pp. 1155--1178, 2007.

\bibitem{slp_book}
D.~Jurafsky and J.~H. Martin,
\newblock {\em Speech and Language Processing},
\newblock Prentice Hall, 2020.

\bibitem{li2020time}
Z.~Li and Z.~Ding,
\newblock ``Time-varying multi-objective optimisation over switching graphs via
  fixed-time consensus algorithms,''
\newblock {\em International Journal of Systems Science}, vol. 51, no. 15, pp.
  2793--2806, 2020.

\bibitem{ning2017distributed}
B.~Ning, Q.~Han, and Z.~Zuo,
\newblock ``Distributed optimization for multiagent systems: An edge-based
  fixed-time consensus approach,''
\newblock {\em IEEE transactions on cybernetics}, vol. 49, no. 1, pp. 122--132,
  2019.

\end{thebibliography}

\end{document}